\documentclass[preprint,12pt]{elsarticle}
\usepackage[latin9]{inputenc}
\usepackage{amsmath}
\usepackage{amsthm}
\usepackage{amssymb}
\usepackage{esint}
\usepackage{setspace}
\usepackage{esint}
\usepackage{textcomp}
\usepackage{mathtools}

\makeatletter

\usepackage{amssymb}

\newcommand{\E}{{\mathbb E}}
\usepackage{soul}

\def\I{I\!\!I}
\newcommand{\N}{\mathbb{N}}
\newcommand{\eps}{\varepsilon}

\makeatletter

\theoremstyle{plain}
\newtheorem{theorem}{\protect\theoremname}
\newtheorem{lem}{\protect\lemmaname}

\newtheorem{prop}{\protect\prop}

\renewcommand{\i}{\textrm{i}}
\providecommand{\theoremname}{Theorem}
\providecommand{\remname}{Remark}
\providecommand{\lemmaname}{Lemma}
\providecommand{\example}{Example}
\providecommand{\prop}{Proposition}
\providecommand{\corr}{Corollary}
\providecommand{\defname}{Definition}

\newcommand{\A}{\mathcal{A}}
\renewcommand{\P}{\mathbb{P}}

\newcommand{\R}{{\mathbb R}}

\newcommand{\Cov}{\operatorname{cov}}
\newcommand{\Var}{\operatorname{Var}}

\renewcommand{\L}{\mathcal{L}}
\newcommand{\F}{\mathcal{F}}
\newcommand{\K}{\mathcal{K}}

\renewcommand{\Re}{\operatorname{Re}}
\renewcommand{\Im}{\operatorname{Im}}

\global\long\def\argmin{\mathrm{argmin}}
\journal{Statistics \& Probability Letters}

\begin{document}

\begin{frontmatter}



\title{Statistical inference for moving-average L{\'e}vy-driven processes: Fourier-based approach}


\author{Denis Belomestny\( ^{1,2} \), Tatiana Orlova\(^{1}\), and Vladimir Panov\( ^{1} \)}

\address{\( ^{1} \) Laboratory of Stochastic Analysis and its Applications \\
             National Research University Higher School of Economics\\
             Shabolovka, 26, 119049 Moscow,   Russia\\         
             and\\   
    \( ^{2} \)   University of Duisburg-Essen\\
Thea-Leymann-Str. 9, 45127 Essen,  Germany      }
\begin{abstract}
We consider a new method of  the semiparametric statistical estimation for the continuous-time moving average L{\'e}vy processes. We derive the convergence rates of the proposed estimators, and show that these rates are optimal in the minimax sense.
\end{abstract}


\end{frontmatter}

\section{Introduction}
Generally speaking, continuous-time L\'{e}vy-driven moving average  processes are defined as 
\begin{eqnarray}\label{Zt}
 Z_t=\int_{-\infty}^\infty \mathcal{K}(t-s)\, dL_s
\end{eqnarray} 
where  $\K$  is a deterministic kernel and  $L=(L_t)_{t\in \mathbb{R}}$ is a two-sided L\'{e}vy process with L{\'e}vy triplet \((\gamma, \sigma, \nu)\). The conditions which guarantee that this integral is well-defined are given in the pioneering work by Rajput and Rosinski \cite{rajput1989spectral}. For instance, if \(\int x^2 \nu(dx) <\infty\), it is sufficient to assume that \(\K \in \L^1(\R) \cap \L^2 (\R)\). Some popular choices for the kernel  are $\K(t)=t^\alpha e^{-\lambda t}1_{[0,\infty)}(t)$ with $\lambda>0$ and $\alpha>-1/2$, (Gamma-kernels, see e.g. Barndorff-Nielsen and Schmiegel  \cite{Barndorff-NielsenSchmiegel2009}), or 
$\K(t)= e^{-\lambda |t|}$  (well-balanced Ornstein-Uhlenbeck process, see Schnurr and Woerner \cite{WD}).
\par
Recently,  Belomestny, Panov and Woerner \cite{BPW} consider the  statistical estimation of the L{\'e}vy measure \(\nu\) from the low-frequency observations of the process \((Z_t)\). The approach presented in  \cite{BPW} is rather general - in particular, it  works well under various choices of \(\K\). Nevertheless, this approach is based on the superposition of the Mellin and Fourier transforms of the L{\'e}vy measure, and therefore its practical implementation can meet some computational difficulties. 
\par
In this paper, we present  another method, which  essentially uses the fact that  in some cases there exists a direct relation between the characteristic exponent of the process \(L\) and the characteristic function of the process \(Z.\) Therefore, the characteristic exponent can be estimated from the observations of the process \(Z,\) and further application of the Fourier techniques from Belomestny and Reiss  \cite{BR} and Panov \cite{panovdiss} leads to the construction of a consistent estimator of the L{\'e}vy triplet. 

The paper is organised as follows. In the next session, we provide the specifications of our model. In Section~\ref{mainidea}, we present the key mathematical idea, 
which lies in the core of the estimation procedure presented in Section~\ref{estproc}. The upper and lower error bounds for the proposed estimates are given in Section~\ref{errorbounds}. Joint consideration of the corresponding results, Theorems~\ref{thm1} and \ref{thm2}, yields the optimality of the estimates. 
Finally, in Section~\ref{num}, we illustrate our approach with some numerical examples. All proofs are collected in Section~\ref{proofs}.


\section{Set-up}
\label{setup}
In this work, we consider the integrals of the form \eqref{Zt}, where \(\K\) is a symmetric kernel of the form:
\begin{eqnarray}
\label{Kalpha}
\K_{\alpha}(x):=\left(1-\alpha |x| \right)^{\frac{1}{\alpha}},\quad |x|\leq\alpha^{-1}
\end{eqnarray}
for some $\alpha\in(0,1).$ As a limiting case for $\alpha\searrow0,$
we get the exponential kernel $\K_{0}(x)=\exp(-x).$ 
Here, for simplicity, we restrict our attention to a particular class of two-sided L{\'e}vy processes with jumps represented by a compound Poisson process \(CPP_t\),
\begin{eqnarray}
\label{model}
	L_{t} = \gamma t + \sigma W_t + CPP^{(1)}_{t} \cdot \I\left\{
	t \geq 0
\right\}
+
CPP^{(2)}_{t} \cdot \I\left\{
	t < 0
\right\},
\end{eqnarray}
\begin{eqnarray}\label{model2}
CPP^{(k)}_{t} := \sum_{j=1}^{N^{(k)}_{t}} Y^{(k)}_{j}, \qquad k=1,2,
\end{eqnarray}
where  \(
 \gamma \in \R
\) is a drift, \(
 \sigma \geq 0
\), \(
 W_t
\) is a Brownian motion, 
\(N_{t}^{(1)}, N_{t}^{(2)},\) are 2 Poisson processes with intensity \(\lambda\),  \(Y_{1}^{(1)}, Y_{2}^{(1)}, ...\) and \(Y_{1}^{(2)}, Y_{2}^{(2)}, ...\) are i.i.d. r.v's with absolutely continuous distribution, and  all \(Y\)'s, \(N^{(1)}_{t},\) \(N^{(2)}_{t}\), \(W_t
\) are jointly independent.   
Due to the L\'evy-Khintchine formula, the characteristic exponent of  $L$  is given by
\begin{eqnarray}
\nonumber
\psi(u)=\log\E\left[
  e^{
\i u L_{1}	
}
\right]
&=&
\i \gamma u -\frac{1}{2}\sigma^{2} u^{2}+\int_\R \left(
  e^{\i u x}-1
 \right) \nu(dx)
\label{LK1}
\\
&=&\i \gamma u -\frac{1}{2}\sigma^{2} u^{2}-\lambda+\mathcal{F}[\nu](u),\label{LK1}
\end{eqnarray}
where \(\nu\) is a L{\'e}vy measure of \((L_t)\), and 
$\mathcal{F}[\nu](u)=\int_{\R}e^{iux}\nu(dx)$
stands for the Fourier transform of $\nu.$

It is important to note that the process  $\left(Z_{t}\right)_{t\in\mathbb{R}}$ is strictly stationary with the characteristic function of the form
\begin{eqnarray}\label{Phi}
\Phi(u)  :=  \E\left[e^{\i uZ_{t}}\right] =\exp\left(\Psi(u)\right),
\quad \mbox{where} \quad
\Psi(u) := \int_{\R}\psi(u\,\mathcal{K}(s))\, ds,
\end{eqnarray}
and therefore for any time points \(t_1, ..., t_n,\) the r.v.'s \(Z_{t_1}, ..., Z_{t_n}\) are identically distributed (but dependent). Our aim  is to estimate the L{\'e}vy triplet \((\gamma, \sigma, \nu)\) based on  the equidistant observations of the process \(Z_t\) at the time points \(\Delta, 2\Delta, .., n \Delta,\) where \(\Delta>0\) is fixed (low-frequency set-up).
 \section{Main idea}
 \label{mainidea}
 The key observation is that under our choice of the kernel function \(\K,\) we can represent the characteristic exponent \(\psi(\cdot)\) of the process \((L_t)\) via the characteristic function \(\Phi(\cdot)\) of the process \(Z_t\). More precisely, since 
\begin{eqnarray*}
\K'_{\alpha}(x)=-\left(1-\alpha x\right)^{\frac{1-\alpha}{\alpha}}=-\K_{\alpha}^{1-\alpha}(x),
\end{eqnarray*} we have
\begin{eqnarray}
\Phi(u) & = & \exp\left[2\int_{0}^{1/\alpha}\psi(u\K_{\alpha}(x))\,dx\right]
\nonumber\\
 & = & \exp\left[2\int_{0}^{1}\psi(uy)y{}^{\alpha-1}\,dy\right]
 =
  \exp\left[2\,u^{-\alpha}\int_{0}^{u}\psi(z)z{}^{\alpha-1}\,dz\right]. \label{Phiu}
\end{eqnarray}
Therefore,  we derive
\begin{eqnarray}\label{psiu}
\psi(u)=\frac{1}{2}u^{1-\alpha}\left(u^{\alpha}\log(\Phi(u))\right)'=\frac{1}{2}\left(\alpha\log(\Phi(u))+u\frac{\Phi'(u)}{\Phi(u)}\right),
\end{eqnarray}
since $\psi(u)u^{\alpha-1}\to0$ as $u\to+0$ provided that  \(\int |x| \nu(dx) < \infty\), see Lemma~\ref{lem1}. Therefore, the characteristic exponent $\psi$ can be directly estimated from
data via a plug-in estimator based on the empirical characteristic function of \(Z\).

Moreover, returning to the representation \eqref{LK1}, we conclude that the L{\'e}vy triplet  \((\gamma, \sigma, \nu)\) can be estimated from \(\psi.\)
In fact,  since $\nu$ is absolutely
continuous with an absolutely integrable density, then by the Riemann-Lebesgue
lemma (see \cite{Kawata}, p. 43) $\mathcal{F}[\nu](u)\to0$
as $|u|\to\infty,$ and consequently $\psi(u)$ can be viewed, at least
for large $|u|,$ as a second order polynomial with the coefficients 
$(-\lambda,i\gamma,-\sigma^{2}/2).$  This observation gives rise for the estimation procedure, which we present in the next session.

\section{Estimation procedure}
\label{estproc}
Assume that the process $(Z_{t})$ is observed on the equidistant time
grid $t\in\left\{ \Delta,2\Delta,...,n\Delta\right\} ,$ where $\Delta$
is fixed.
\newline

\textbf{Step 1: estimation of \(\psi\).} Define 
\[
\Phi_{n}(u)=\frac{1}{n}\sum_{j=1}^{n}e^{iuZ_{j\Delta}},
\]
and set 
\[
\psi_{n}(u)=\frac{1}{2}\left(\alpha \log(\Phi_{n}(u))+u\frac{\Phi_{n}'(u)}{\Phi_{n}(u)}\right),
\]
where the branch of the complex logarithm is taken in such a way that
$\psi_{n}$ is continuous on $(-x_{0,n},x_{0,n})$ with $\psi_{n}(0)=0$
and $x_{0,n}$ being the first zero of $\Phi_{n}.$ In fact, since
$\Phi$ does not vanish on $\mathbb{R}$, we have $x_{0,n}\stackrel{a.s.}{\to}\infty.$
\newline

\textbf{Step 2: estimation of \(\sigma\) and \(\lambda\).}  Let $U_{n}\to\infty$ and 
\[
\widetilde{w}^{U_{n}}(u):=(1/U_{n})\;\widetilde{w}\left(u/U_{n}\right),
\]
where $\widetilde{w}(u)$ is a continuous function, supported on the
interval $[\eps,1]$ with some \(\eps>0.\) 
Consider now the optimisation problem 
\begin{equation}
(\sigma_{n}^{2},\lambda_{n}):=\argmin_{(\sigma^{2},\lambda)}\int_{0}^{\infty}\widetilde{w}^{U_{n}}(u)(\Re[\psi_{n}(u)]+\sigma^{2}u^{2}/2+\lambda)^{2}\,du,\label{eq:weight_ls}
\end{equation}
which has the solution 
\begin{eqnarray}
\sigma_{n}^{2}=\int_{0}^{\infty}w_{\sigma}^{U_{n}}(u)\Re\psi_{n}(u)\,du,\label{sigma_est_opt}
\end{eqnarray}
with
\begin{eqnarray}
w_{\sigma}^{U_{n}}(u) & := & \widetilde{w}^{U_{n}}(u)\frac{2\left[\left(\int_{0}^{\infty}\widetilde{w}^{U_{n}}(s)\:ds\right)u^{2}-\int_{0}^{\infty}\widetilde{w}^{U_{n}}(s)s^{2}\:ds\right]}{\left(\int_{0}^{\infty}\widetilde{w}^{U_{n}}(s)s^{2}\:ds\right)^{2}-\int_{0}^{\infty}\widetilde{w}^{U_{n}}(s)s^{4}\:ds\;\cdot\;\int_{0}^{\infty}\widetilde{w}^{U_{n}}(s)\:ds}.
\nonumber
\\ \label{eq: wsigma}
\end{eqnarray}
Note that the weighting function $w_{\sigma}^{U_{n}}(u)$ satisfies the property $w_{\sigma}^{U_{n}}(u)=U_{n}^{-3}w_{\sigma}^{1}(u/U_{n})$, and moreover, 
\begin{equation}
\int_{0}^{U_{n}}(-u^{2}/2)w_{\sigma}^{U_{n}}(u)\,du=1,\quad\int_{0}^{U_{n}}w_{\sigma}^{U_{n}}(u)\,du=0.\label{wsigma_prop}
\end{equation}
Analogously, 
\begin{eqnarray}
\lambda_{n}=\int_{0}^{\infty}w_{\lambda}^{U_{n}}(u)\Re\psi_{n}(u)\,du\label{lambda_est_opt}
\end{eqnarray}
 holds with $w_{\lambda}^{U_{n}}(u)=U_{n}^{-1}w_{\lambda}^{1}(u/U_{n})$
satisfying the properties 
\[
\int_{0}^{U_{n}}(-1)w_{\lambda}^{U_{n}}(u)\,du=1,\quad\int_{0}^{U_{n}}(-u^{2}/2)w_{\lambda}^{U_{n}}(u)\,du=0.
\]
\newline

\textbf{Step 3: estimation of \(\gamma\).} Finally,  the  parameter \(\gamma\) can be estimated by considering the optimisation problem 
\begin{eqnarray}
\gamma_{n}:=\argmin_{\gamma}\int_{0}^{\infty}\widetilde{w}^{U_{n}}(u)(\Im\psi_{n}(u)-\gamma u)^{2}\,du,\label{weight_ls1}
\end{eqnarray}
which leads to the estimate 
\begin{eqnarray}
\gamma_{n}=\int_{0}^{\infty}w_{\gamma}^{U_{n}}(u)\Im\psi_{n}(u)\,du,\label{gamma_est_opt}
\end{eqnarray}
where $w_{\gamma}^{U_{n}}(u)=U_{n}^{-2}w_{\gamma}^{1}(u/U_{n})$ fulfills
\(
\int_{0}^{U_{n}}u\,w_{\gamma}^{U_{n}}(u)\,du=1.
\)
All functions $w_{\sigma}^{1}$, $w_{\gamma}^{1}$ and $w_{\lambda}^{1}$
are supported on $[\eps,1]$ and bounded. 
\newline

\textbf{Step 4: estimation of the L{\'e}vy density.}
Note that under our assumptions on the L{\'e}vy process \((L_t)\) (see Section~\ref{setup}),  the Levy measure
$\nu$ possesses a density, which we denote, with a slight abuse of notation, also by $\nu(x)$. This L{\'e}vy density can be estimated as 
a regularised inverse Fourier transform of the remainder: 
\begin{equation}
\nu_{n}(x):={\cal F}^{-1}\left[\Bigl(\psi_{n}(\cdot)+\tfrac{\sigma_{n}^{2}}{2}(\cdot)^{2}-i\gamma_{n}(\cdot)+\lambda_{n}\Bigr)w_{\nu}(\cdot/U_{n})\right](x),\quad x\in\R,\label{nuhatdef}
\end{equation}
where $w_{\nu}$ is a weight function supported on $[-1,1].$ Note
that $\int_{\R}\nu_{n}(x)\,dx=\lambda_{n},$ if $w_{\nu}(0)=1.$

\section{Error bounds}
\label{errorbounds}
\begin{theorem}
\label{thm1}
Consider the model~\eqref{Zt}, where \(\K\) is a kernel in the form \eqref{Kalpha} and  \((L_t)\) is a L{\'e}vy process in the form~\eqref{model} with triplet \(\left(
  \gamma, \sigma, \nu
\right)\). Assume that the L{\'e}vy density \(\nu\) is \(s\)-times weakly differentiable  for some \(s \in \N\), and moreover the L{\'e}vy triplet belongs to the class
\begin{eqnarray*}
\mathcal{T}_s = \mathcal{T}_s (\sigma^\circ, R)=\Biggl\{
\sigma \in (0, \sigma^\circ), \;\;
\int x^2 \nu(dx) \leq R, \;\;
\left\Vert \nu^{(s)}\right\Vert _{\infty}\leq R
\Biggr\}
\end{eqnarray*}
with some \(\sigma^\circ, R >0.\)
 Assume also that the weighting functions satisfy the conditions 
\begin{eqnarray}
\label{Fw}
\|\mathcal{F}(w_{\sigma}^{1}(u)/u^{s})\|_{L^{1}} < \infty,
\quad
\|\mathcal{F}(w_{\lambda}^{1}(u)/u^{s})\|_{L^{1}}<\infty,
\end{eqnarray}
\begin{eqnarray}
\|\mathcal{F}(w_{\gamma}^{1}(u)/u^{s})\|_{L^{1}} <\infty.
\end{eqnarray}
Then 
 it holds 
\begin{eqnarray*}
\lim_{A \to +\infty}
\varlimsup_{n \to +\infty}
\sup_{\left(
  \gamma, \sigma, \nu
\right) \in \mathcal{T}_s
}
\P \left\{ 
  \left|
   \sigma_n^2 - \sigma^2
\right|
\geq A \cdot U_n^{-(s+3)}
\right\} 
&=&0,\\
\lim_{A \to +\infty}
\varlimsup_{n \to +\infty}
\sup_{\left(
  \gamma, \sigma, \nu
\right) \in \mathcal{T}_s
}
\P \left\{ 
  \left|
   \gamma_n - \gamma
\right|
\geq A \cdot U_n^{-(s+2)}
\right\} 
&=&0,\\
\lim_{A \to +\infty}
\varlimsup_{n \to +\infty}
\sup_{\left(
  \gamma, \sigma, \nu
\right) \in \mathcal{T}_s
}
\P \left\{ 
  \left|
   \lambda_n - \lambda
\right|
\geq A \cdot U_n^{-(s+1)}
\right\} 
&=&0,
\end{eqnarray*}
provided  \(U_n=\sqrt{\kappa \log(n)}\) with some constant \(\kappa>0\) depending on \(\sigma^\circ\) and \(R.\)

\end{theorem} 

As shown in the next theorem, the above rates are optimal in minimax sense.

\begin{theorem}
\label{thm2}
For any \(\sigma^\circ, R >0, \) there exists some \(A>0\) such that 
\begin{eqnarray*}
\varliminf_{n \to +\infty}
\inf_{\breve\sigma_n}
\sup_{\left(
  \gamma, \sigma, \nu
\right) \in \mathcal{T}_s
}
\P \left\{ 
  \left|
   \breve\sigma_n^2 - \sigma^2
\right|
\geq A \cdot \left(
  \log(n) 
\right)^{-(s+3)/2}
\right\} 
&>&0,\\
\varliminf_{n \to +\infty}
\inf_{\breve\gamma_n}
\sup_{\left(
  \gamma, \sigma, \nu
\right) \in \mathcal{T}_s
}
\P \left\{ 
  \left|
   \breve\gamma_n - \gamma
\right|
\geq A \cdot  \left(
  \log(n) 
\right)^{-(s+2)/2}
\right\} 
&>&0,\\
\varliminf_{n \to +\infty}
\inf_{\breve\lambda_n}
\sup_{\left(
  \gamma, \sigma, \nu
\right) \in \mathcal{T}_s
}
\P \left\{ 
  \left|
   \breve\lambda_n - \lambda
\right|
\geq A \cdot  \left(
  \log(n) 
\right)^{-(s+1)/2}
\right\} 
&>&0.
\end{eqnarray*}
where the infimums are taken over  
all possible estimates \(\breve\sigma_n, \breve\gamma_n, \breve\lambda_n\) of the parameters \(\sigma,\gamma,\lambda,\) and supremums - over all triplets  from the class \(\mathcal{T}_s=\mathcal{T}_s(\sigma^\circ, R).\)
\end{theorem}

\section{Numerical example}\label{num}
Consider the integral \eqref{Zt} with the kernel \(\K(x)\) from the class \eqref{Kalpha}, and the L{\'e}vy process \((L_t)\)  defined by \eqref{model}-\eqref{model2}. For simulation study, we take $\gamma=5,$ $\lambda=1$
and $\sigma=0$, and aim to estimate these parameters  under different choices of the parameter $\alpha$, namely $\alpha=0.5,$
$0.8$ and $0.9$.

\textbf{Simulation.}
For $k=1,2$, denote the jump times of $L_{t}^{(k)}$ by
$s_{1}^{(k)},s_{2}^{(k)},....$, corresponding to the jump sizes $Y_{1}^{(k)},Y_{2}^{(k)},...$  Note that 
$$
Z_{t}=
\begin{cases} 
\frac{2\gamma}{1+\alpha}+\underset{k\in K^{(1)}}{\sum}\left(1-\alpha|t-s_{k}^{(1)}|\right)^{1/\alpha}Y_{k}^{\left(1\right)}, &  
\text{if} \;\; t \geq \frac{1}{\alpha}\\
\frac{2\gamma}{1+\alpha}+\underset{k\in K^{(2)}}{\sum}\left(1-\alpha|t-s_{k}^{(1)}|\right)^{1/\alpha}Y_{k}^{\left(1\right)}&\\
\hspace{2.5cm} +\underset{k\in K^{(3)}}{\sum}\left(1-\alpha|t+s_{k}^{(2)}|\right)^{1/\alpha}Y_{k}^{\left(2\right)},
 & \text{if} \;\; t<\frac{1}{\alpha},
\end{cases}
$$
where\begin{eqnarray*}
K^{\left(1\right)} &:=& \left\{ k:t-\frac{1}{\alpha}\leq s_{k}^{(1)}\leq t+\frac{1}{\alpha}\right\}, \\
K^{\left(2\right)} &:=&  \left\{ k:0\leq s_{k}^{(1)}\leq t+\frac{1}{\alpha}\right\}, \\
K^{\left(3\right)} &:=& \left\{ k:0\leq s_{k}^{(2)}\leq\frac{1}{\alpha}-t\right\} .
\end{eqnarray*}

\noindent Typical trajectory of the process $Z_{t}$ is presented
on Figure~\ref{fig1}. 

\begin{figure}
\includegraphics[width=20cm,height=5cm,keepaspectratio]{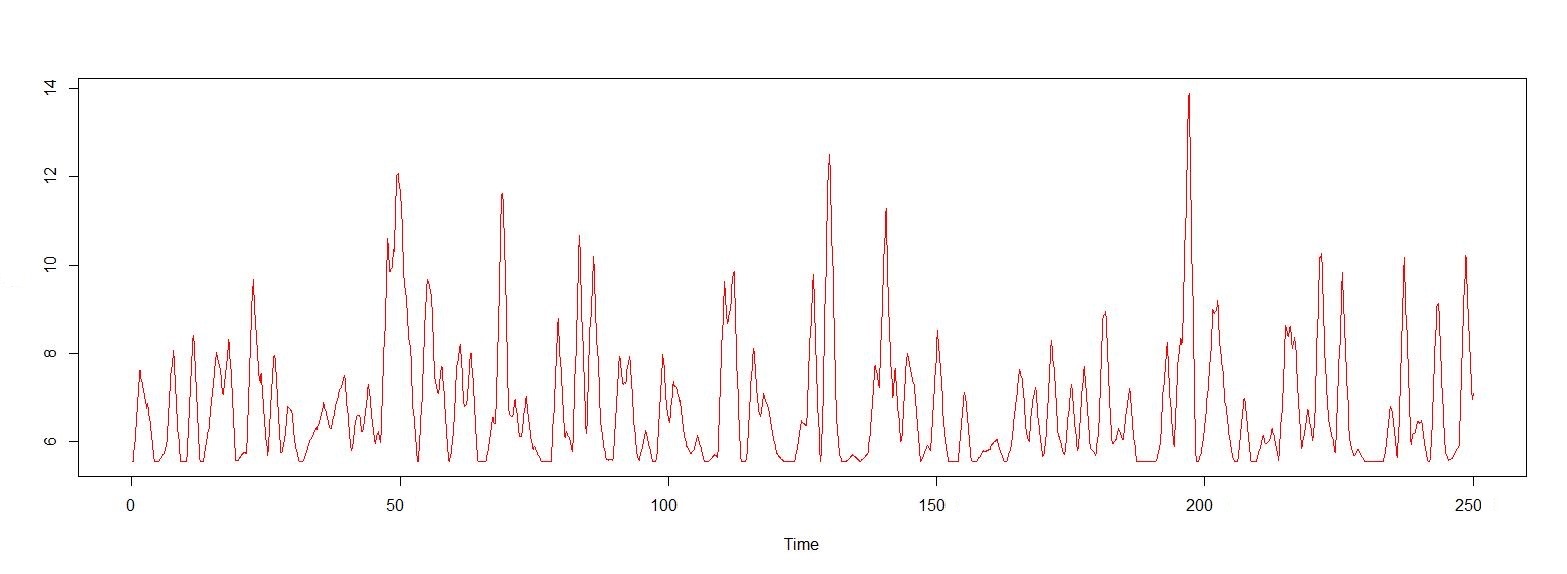}\protect\caption{\label{fig1}Typical trajectory of the process $Z_{t}$ with value of the parameter
$\alpha=0.8$}
\end{figure}

\textbf{Estimation.} Following the ideas from Section~\ref{estproc}, we estimate the parameters 
\(\gamma, \lambda, \sigma\) under different choices of \(\alpha.\)

To show the convergence properties of the considered estimates, we provide simulations
with different values of \(n\). The boxplots of the corresponding estimation errors (differences) based on
25 simulation runs are presented on Figures~\ref{fig2}, \ref{fig3} and \ref{fig4}.
Note that  the parameter $U_{n}$ is chosen by numerical optimisation.  The exact values are presented in Tables~\ref{table1} and \ref{table2}.
\par
The simulation study illustrates our theoretical results on the rates of convergence given in Section~\ref{errorbounds}. In fact, visual comparison of Figures~\ref{fig2}, \ref{fig3} and \ref{fig4} shows that  the  proposed estimator for the parameter \(\sigma\) has the highest speed of convergence to the true value,  whereas the corresponding speed for \(\gamma_n\) is lower, and for \(\lambda_n\) even more low (cf with the rates in Theorem~\ref{thm1}).
Moreover, the simulations results show that the convergence rates significantly
depend on the parameter $\alpha$. More precisely, it turns out that the quality of estimation increases with growing \(\alpha\), and the best rates correspond to the case when $\alpha$ is close to 1. This can explained by the fact that observations become less independent as \(\alpha\) increases.

\noindent 
\begin{figure}
\includegraphics[width=18cm,height=5cm,keepaspectratio]{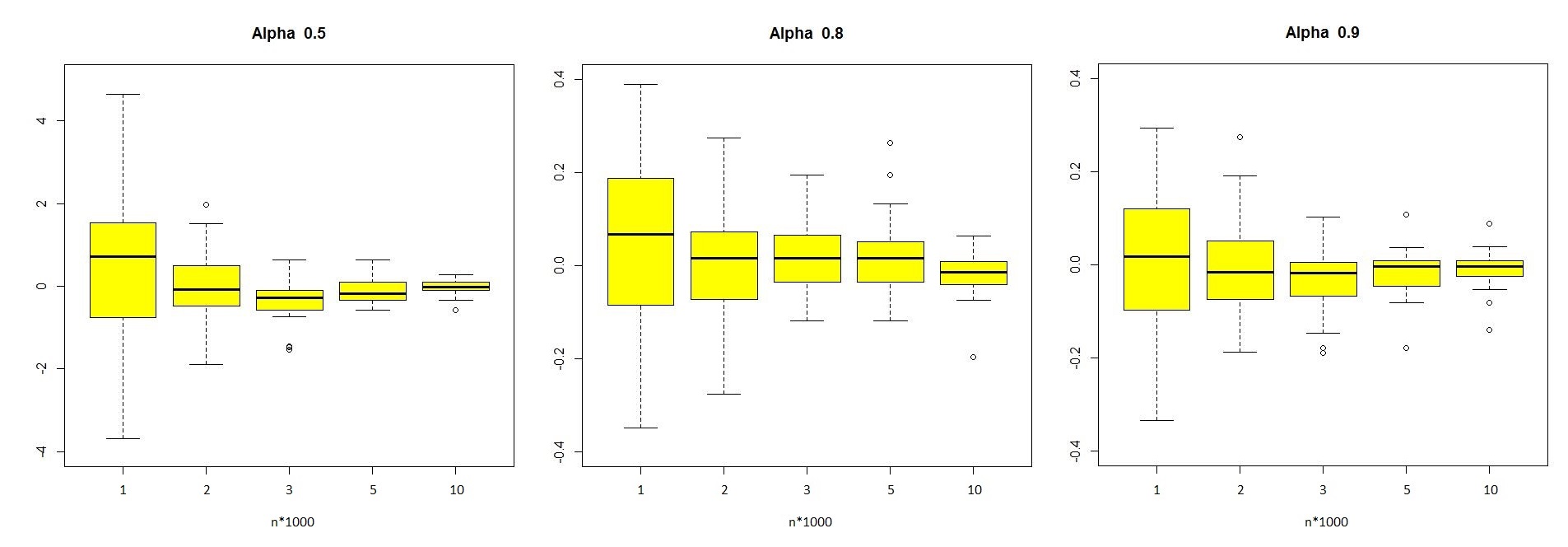}\protect\caption{\label{fig2} Boxplots $\hat{\gamma}-\gamma$ based on 25 simulation runs}

\end{figure}

\begin{figure}

\includegraphics[width=18cm,height=5cm,keepaspectratio]{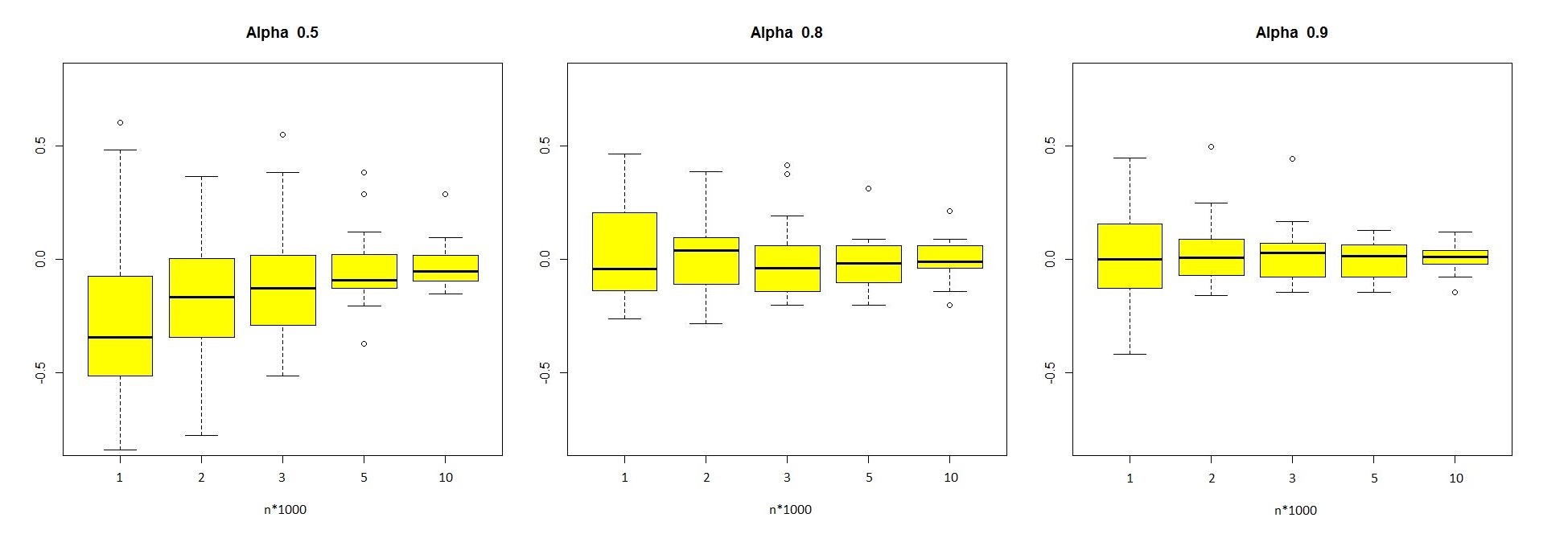}\protect\caption{\label{fig3} Boxplots of $\hat{\lambda}-\lambda$ based on 25 simulation runs}
\end{figure}
\begin{figure}

\includegraphics[width=18cm,height=5cm,keepaspectratio]{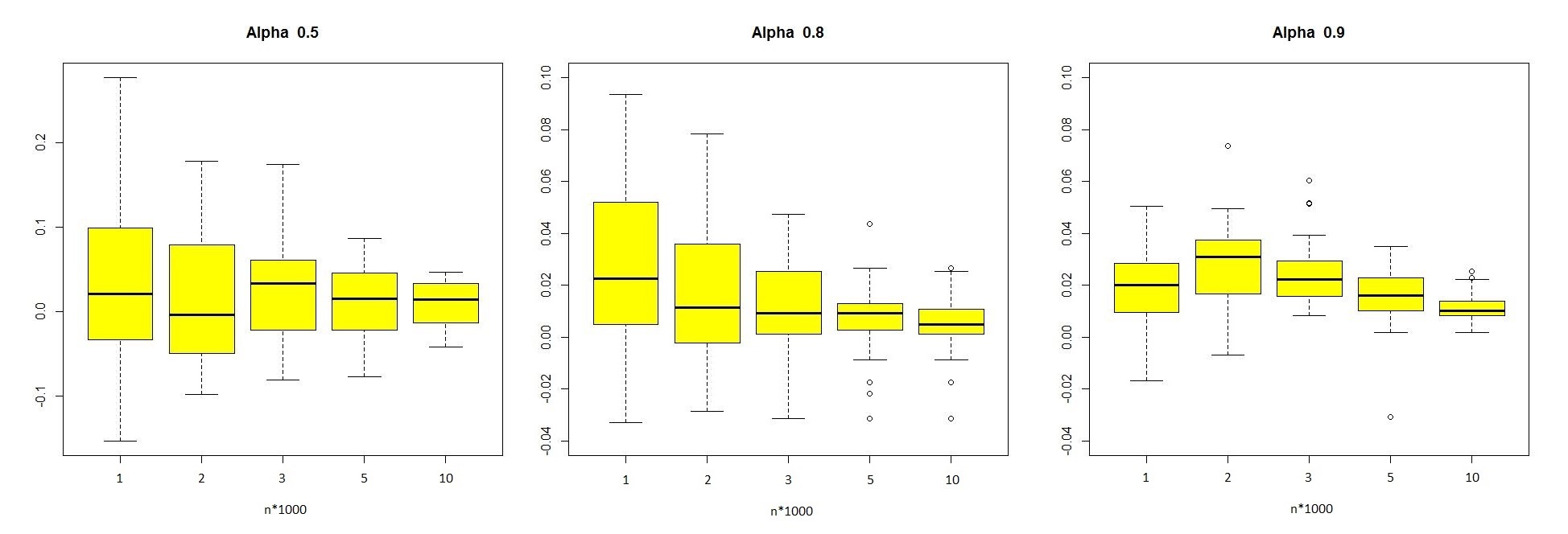}\protect\caption{\label{fig4} Boxplots of $\hat{\sigma^{2}}-\sigma^{2}$ based on 25 simulation
runs}
\end{figure}

\begin{table}
\begin{center}
\begin{tabular}{|c|c|c|}
\hline 
$\alpha=0.5$ & $\alpha=0.8$ & $\alpha=0.9$\tabularnewline
\hline 
\hline 
\begin{tabular}{|c|c|}
\hline 
$n$ & $U_{n}$\tabularnewline
\hline 
\hline 
1000 & 1.2\tabularnewline
\hline 
2000 & 1.35\tabularnewline
\hline 
3000 & 1.4\tabularnewline
\hline 
5000 & 1.45\tabularnewline
\hline 
10000 & 1.55\tabularnewline
\hline 
\end{tabular} & %
\begin{tabular}{|c|c|}
\hline 
$n$ & $U_{n}$\tabularnewline
\hline 
\hline 
1000 & 3.5\tabularnewline
\hline 
2000 & 3.5\tabularnewline
\hline 
3000 & 3.8\tabularnewline
\hline 
5000 & 4\tabularnewline
\hline 
10000 & 4.2\tabularnewline
\hline 
\end{tabular} & %
\begin{tabular}{|c|c|}
\hline 
$n$ & $U_{n}$\tabularnewline
\hline 
\hline 
1000 & 4.5\tabularnewline
\hline 
2000 & 4.5\tabularnewline
\hline 
3000 & 4.6\tabularnewline
\hline 
5000 & 4.8\tabularnewline
\hline 
10000 & 5.1\tabularnewline
\hline 
\end{tabular}\tabularnewline
\hline 
\end{tabular}\protect\caption{\label{table1} Optimal sequences for the estimation of  parameter $\lambda$}
\end{center}
\end{table}
\begin{table}
\begin{center}
\begin{tabular}{|c|c|c|}
\hline 
$\alpha=0.5$ & $\alpha=0.8$ & $\alpha=0.9$\tabularnewline
\hline 
\hline 
\begin{tabular}{|c|c|}
\hline 
$n$ & $U_{n}$\tabularnewline
\hline 
\hline 
1000 & 0.8\tabularnewline
\hline 
2000 & 0.85\tabularnewline
\hline 
3000 & 1\tabularnewline
\hline 
5000 & 1.1\tabularnewline
\hline 
10000 & 1.3\tabularnewline
\hline 
\end{tabular} & %
\begin{tabular}{|c|c|}
\hline 
$n$ & $U_{n}$\tabularnewline
\hline 
\hline 
1000 & 2.7\tabularnewline
\hline 
2000 & 2.75\tabularnewline
\hline 
3000 & 2.75\tabularnewline
\hline 
5000 & 2.8\tabularnewline
\hline 
10000 & 3\tabularnewline
\hline 
\end{tabular} & %
\begin{tabular}{|c|c|}
\hline 
$n$ & $U_{n}$\tabularnewline
\hline 
\hline 
1000 & 3\tabularnewline
\hline 
2000 & 3.2\tabularnewline
\hline 
3000 & 3.2\tabularnewline
\hline 
5000 & 3.3\tabularnewline
\hline 
10000 & 3.5\tabularnewline
\hline 
\end{tabular}\tabularnewline
\hline 
\end{tabular}\protect\caption{\label{table2}  Optimal sequences for the estimation of  parameter $\gamma$}
\end{center}
\end{table}
\begin{table}
\begin{center}
\begin{tabular}{|c|c|c|}
\hline 
$\alpha=0.5$ & $\alpha=0.8$ & $\alpha=0.9$\tabularnewline
\hline 
\hline 
\begin{tabular}{|c|c|}
\hline 
$n$ & $U_{n}$\tabularnewline
\hline 
\hline 
1000 & 8\tabularnewline
\hline 
2000 & 8\tabularnewline
\hline 
3000 & 8\tabularnewline
\hline 
5000 & 8.2\tabularnewline
\hline 
10000 & 8.5\tabularnewline
\hline 
\end{tabular} & %
\begin{tabular}{|c|c|}
\hline 
$n$ & $U_{n}$\tabularnewline
\hline 
\hline 
1000 & 8.5\tabularnewline
\hline 
2000 & 8.55\tabularnewline
\hline 
3000 & 8.6\tabularnewline
\hline 
5000 & 8.7\tabularnewline
\hline 
10000 & 8.8\tabularnewline
\hline 
\end{tabular} & %
\begin{tabular}{|c|c|}
\hline 
$n$ & $U_{n}$\tabularnewline
\hline 
\hline 
1000 & 8.75\tabularnewline
\hline 
2000 & 8.8\tabularnewline
\hline 
3000 & 8.8\tabularnewline
\hline 
5000 & 9\tabularnewline
\hline 
10000 & 9.2\tabularnewline
\hline 
\end{tabular}\tabularnewline
\hline 
\end{tabular}\protect\caption{Optimal sequences for the estimation of  parameter $\sigma$}
\end{center}
\end{table}

\newpage
\section{Proofs}
\label{proofs}
\subsection{Proof of Theorem~\ref{thm1}}

\textbf{1.} 
For the sake of clarity we focus our analysis on the estimate $\sigma_{n}.$
First note that by (\ref{sigma_est_opt}) and (\ref{wsigma_prop})
the difference $\sigma_{n}^{2}-\sigma^{2}$ can be decomposed as follows:

\begin{eqnarray}
\sigma_{n}^{2}-\sigma^{2} & = & \int_{0}^{U_{n}}w_{\sigma}^{U_{n}}(u)\Re\left(\psi_{n}(u)-\psi(u)\right)\,du \nonumber\\
&&\hspace{5.5cm}+\int_{0}^{U_{n}}w_{\sigma}^{U_{n}}(u)\Re\psi(u)\,du-\sigma^{2}\nonumber \\
 & = & \underbrace{\int_{0}^{U_{n}}w_{\sigma}^{U_{n}}(u)\Re\left(\psi_{n}(u)-\psi(u)\right)\,du}_{\text{Statistical error}}+\underbrace{\int_{0}^{U_{n}}w_{\sigma}^{U_{n}}(u)\Re \mathcal{F}[\nu](u)\,du}_{\text{ Bias }}.
 \nonumber\\
 &&\label{eq:error_dec}
\end{eqnarray}

\textbf{2.} Let us first consider the bias term in (\ref{eq:error_dec}). Note that  its order obviously depends on the decay of the Fourier transform
$\mathcal{F}[\nu](u),$ which is related to the smoothness
of $\nu$, see \cite{Kawata}.  Then by the Plancherel identity 
\begin{eqnarray}
\nonumber
\left|
\int_{0}^{U_{n}}w_{\sigma}^{U_{n}}(u)\Re \mathcal{F}[\nu](u)\,du
\right|
& \leq &
\left|\int_{0}^{\infty}w_{\sigma}^{U_{n}}(u)\mathcal{F}[\nu](u)du\right| \\
\nonumber
& = & 2\pi\,\left|\int_{-\infty}^{\infty}\nu^{(s)}(x)\overline{{\cal F}^{-1}[w_{\sigma}^{U_{n}}(\cdot)/(i\cdot)^{s}](x)}dx\right|\\
\nonumber
 & \leq & U_n^{-(s+3)}\|\nu^{(s)}\|_{\infty}\|\mathcal{F}(w_{\sigma}^{1}(u)/u^{s})\|_{L^{1}}\lesssim U_{n}^{-(s+3)},\\
 &&\label{bias}
\end{eqnarray}
since $w_{\sigma}^{U_{n}}(u)=U_{n}^{-3}w_{\sigma}^{1}(u/U_{n})$, \(\left\Vert \nu^{(s)}\right\Vert _{\infty}\le R,\) and \eqref{Fw}.

\textbf{2.} As for the statistical error, we first note that 
\begin{eqnarray*}
\psi_{n}(u)-\psi(u) &=& 
\frac{
\alpha
}{2
}
 \left(
  \log\left(
  \Phi_n(u)
\right) - \log\left(
  \Phi(u)
\right)
\right)
+
\frac{
u
}{2
}
\left(
\frac{\Phi_{n}'(u)}{\Phi_{n}(u)}-\frac{\Phi'(u)}{\Phi(u)}  
\right).
\end{eqnarray*}
Consider the event \begin{eqnarray*}
\mathcal{B}_{n,A}:=\left\{\max_{j=0,1}\sup_{u \in [-U_n, U_n]} \left|D_{n,j}(u)\right|\leq  A \eps_n \right\},
\end{eqnarray*}
where  $D_{n,j}(u)=(\Phi_{n}^{(j)}(u)-\Phi^{(j)}(u))/\Phi(u),$ $j=0,1,$ and \(\eps_n \to 0\) as \(n \to \infty.\) Using the same techniques as in Theorem~2 from \cite{BPW}, one can show that from the condition \(\int_{|x|>1} x^2 \nu(dx) <\infty,\) if follows that 
\begin{eqnarray*}
\P\Bigl\{\mathcal{B}_{n,A}\Bigr\}\geq  
1 - 
\frac{
C_{1}
}{
	\sqrt{A}
}
\frac{
	\sqrt{U_{n}} n^{(1/4) - C_{1}A^{2}}
}
{
	\log^{1/4}(n)
}
,
\end{eqnarray*}
provided \(\eps_n=\sqrt{\log(n)/n} \exp\{C_2 \sigma^2 U_n^2\}\) with some \(C_1, C_2>0\) depending on \(\alpha.\) On the event \(\mathcal{B}_{n,A}\), it holds
\begin{eqnarray*}
\log\left(
  \Phi_n(u)
\right) - \log\left(
  \Phi(u)
\right)
= D_{n,0}(u) + O\left(|D_{n,0}(u)|^2\right),
\end{eqnarray*}
because \(
\left| \log(1+z) -z \right| \leq 2 |z|^2\) for any \(|z|<1/2.\)
Moreover,
\begin{align*}
\frac{\Phi_{n}'(u)}{\Phi_{n}(u)}-\frac{\Phi'(u)}{\Phi(u)} & =\frac{\Phi(u)\Phi_{n}'(u)-\Phi'(u)\Phi_{n}(u)}{\Phi_{n}(u)\Phi(u)}\\
 & =\frac{(\Phi(u)-\Phi_{n}(u))\Phi_{n}'(u)+(\Phi_{n}'(u)-\Phi'(u))\Phi_{n}(u)}{\Phi_{n}(u)\Phi(u)}\\
 & =-\frac{\Phi_{n}'(u)}{\Phi_{n}(u)}D_{n,0}(u)+D_{n,1}(u)\\
 & =\left[\frac{\Phi'(u)}{\Phi(u)}-\frac{\Phi_{n}'(u)}{\Phi_{n}(u)}\right]D_{n,0}(u)-\frac{\Phi'(u)}{\Phi(u)}D_{n,0}(u)+D_{n,1}(u),
\end{align*}
and therefore on the event \(\A,\)\begin{eqnarray*}
\frac{\Phi_{n}'(u)}{\Phi_{n}(u)}-\frac{\Phi'(u)}{\Phi(u)}&=&\frac{D_{n,1}(u)-\frac{\Phi'(u)}{\Phi(u)}D_{n,0}(u)}{1+D_{n,0}(u)}\\
&=& D_{n,1}(u)-\frac{\Phi'(u)}{\Phi(u)}D_{n,0}(u)\\
&&\hspace{3cm}+O\left(|D_{n,0}(u)|^2\right)+
O\left(|D_{n,0}(u)D_{n,1}(u)|\right),
\end{eqnarray*}
where we use the inequality \(|((1+z)^{-1} -1|\leq 2 |z|\) for any \(|z| < 1/2.\)
Therefore,  the statistical error can be further decomposed as follows:
\begin{eqnarray*}
\int_{0}^{U_{n}}w_{\sigma}^{U_{n}}(u)\Re\bigl(\psi_{n}(u)-\psi(u)\bigr)\,du =\frac{
1
}{2
}\left(
  L_n
+R_{n}
\right)
\end{eqnarray*}
with the first order
(linear) term \(L_n=\Re
   \breve{L}_n,
\)
\begin{eqnarray*}
\breve{L}_{n} & := & 
\int_{0}^{U_{n}}w_{\sigma}^{U_{n}}(u)
\left( 
  \alpha - u \frac{\Phi'(u)}{\Phi(u)} 
 \right) D_{n,0}(u)\,du\\
& &\hspace{5cm}+\int_{0}^{U_{n}}w_{\sigma}^{U_{n}}(u)uD_{n,1}(u)\,du
\end{eqnarray*}
and the remainder $R_{n}$, which contains higher order powers of $D_{n,0}$ and $D_{n,1}.$
On the event \(\mathcal{B}_{n,A},\)
\begin{eqnarray*}
|R_n| \leq \max\Bigl(\max_u \left|D_{n,0}\right|^2, \max_u \left|D_{n,0}(u) D_{n,1}(u)\right|
\Bigr)
\int_0^{U_n} w_\sigma^{U_n} (u) du
\lesssim A^2 \frac{
\eps_n^2 
}{
   U_n^{2}
},
\end{eqnarray*}
and we finally conclude that at least for large \(n\) it holds

\begin{eqnarray}\label{Rn}
\P\left\{
|R_n|  > A^2  g_{n,1}
\right\}
\leq
\frac{
C_{1}
}{
	\sqrt{A}
}
\frac{
	\sqrt{U_{n}} n^{(1/4) - C_{1}A^{2}}
}
{
	\log^{1/4}(n)
},\end{eqnarray}
where
\[g_{n,1}=
C_3\frac{
\log(n)
}{
n
}
\frac{
e^{2 C_2 \sigma^2 U_n^2}
}{
U_n^2
}
\]
with some \(C_3>0.\)

 \textbf{4.} 
The linear term \(L_n\) can be analysed as follows. We have $\E[\breve{L}_{n}]=0,$ \; \(\Var L_n \leq \Var \breve{L}_n,\)
and 
\begin{eqnarray*}
\Var\breve{L}_{n} &=& \int_{0}^{U_{n}}\int_{0}^{U_{n}}w_{\sigma}^{U_{n}}(u)w_{\sigma}^{U_{n}}(v)\,\left( 
  \alpha - u \frac{\Phi'(u)}{\Phi(u)} 
 \right)
 \overline{
 \left( 
  \alpha - v \frac{\Phi'(v)}{\Phi(v)} 
 \right)
 }
  \frac{
 1
 }{
 \Phi(u) \overline{\Phi(v)}
 }
 \\
 && \hspace{6cm}\cdot\Cov_{\mathbb{C}}\left(\Phi_{n}(u),\Phi_{n}(v)\right)\,du\,dv\label{eq:VarLn}\\
 &  & + \int_{0}^{U_{n}}\int_{0}^{U_{n}}w_{\sigma}^{U_{n}}(u)w_{\sigma}^{U_{n}}(v)
 \left( 
  \alpha - u \frac{\Phi'(u)}{\Phi(u)} 
 \right)
  \frac{
v
 }{
 \Phi(u) \overline{\Phi(v)}
 }
 \\
  && \hspace{6cm}\cdot
 \Cov_{\mathbb{C}}
 \left(\Phi_{n}(u),\Phi'_{n}(v)\right)\,du\,dv\\
 &&
+ \int_{0}^{U_{n}}\int_{0}^{U_{n}}w_{\sigma}^{U_{n}}(u)w_{\sigma}^{U_{n}}(v)
  \frac{
u
 }{
 \Phi(u) \overline{\Phi(v)}
 }
\overline{
 \left( 
  \alpha - v \frac{\Phi'(v)}{\Phi(v)} 
 \right)
 }\\
   && \hspace{6cm}\cdot
 \Cov_{\mathbb{C}}\left(\Phi'_{n}(u),\Phi_{n}(v)\right)\,du\,dv\\
 &  & \int_{0}^{U_{n}}\int_{0}^{U_{n}}w_{\sigma}^{U_{n}}(u)w_{\sigma}^{U_{n}}(v)\,   \frac{
u v
 }{
 \Phi(u) \overline{\Phi(v)}
 } \Cov_{\mathbb{C}}\left(\Phi'_{n}(u),\Phi'_{n}(v)\right)\,du\,dv\\
 & = & I_{1}+ I_{2}+I_{3} +I_{4}.
 \end{eqnarray*}
It holds 
\[
\Cov_{\mathbb{C}}\left(\Phi_{n}(u),\Phi_{n}(v)\right)=\frac{1}{n^{2}}\sum_{j,k=1}^{n}\mathrm{Cov_{\mathbb{C}}}\left(e^{iuZ_{j\Delta}},e^{ivZ_{k\Delta}}\right).
\]
Let $t>s$ and compute
\begin{eqnarray*}
\mathrm{Cov_{\mathbb{C}}}\left[e^{iuZ_{t}},e^{ivZ_{s}}\right] & = & \mathrm{E}\left[e^{iuZ_{t}-ivZ_{s}}\right]-\mathrm{E}\left[e^{iuZ_{t}}\right]\mathrm{E}\left[e^{-ivZ_{s}}\right]\\
 & = & \Phi_{(Z_{t},Z_{s})}(u,-v)-\Phi_{Z_{t}}(u)\Phi_{Z_{s}}(-v)\\
 & = & \exp\left(\int_{\mathbb{R}}\psi\left(u\K_{\alpha}(x-t)-v\K_{\alpha}(x-s)\right)dx\right)\\
 &  & -\exp\left(\int_{\mathbb{R}}\psi\left(u\K_{\alpha}(x-t)\right)dx\right)\\
 & & \hspace{2cm}\times\exp\left(\int_{\mathbb{R}}\psi\left(-v\K_{\alpha}(x-s)\right)dx\right).
\end{eqnarray*}
Using the inequality $\left|e^{z}-e^{y}\right|\leq\left(\left|e^{z}\right|\vee\left|e^{y}\right|\right)\left|y-z\right|,$
which holds for any $z,y\in\mathbb{C},$ we get
\begin{eqnarray*}
\left|\mathrm{Cov_{\mathbb{C}}}\left[e^{iuZ_{t}},e^{ivZ_{s}}\right]\right| & \leq & \left(\left|\Phi_{(Z_{t},Z_{s})}(u,-v)\right|\vee\left|\Phi_{Z_{t}}(u)\Phi_{Z_{s}}(-v)\right|\right)\\
 &  & \times\biggl|\int_{\mathbb{R}}\Bigl[\psi\left(u\K_{\alpha}(x-t)-v\K_{\alpha}(x-s)\right) \Bigr.\biggr. 
 \\
 && 
 \hspace{2cm}-
\biggl. \Bigl.\psi\left(u\K_{\alpha}(x-t)\right)-\psi\left(-v\K_{\alpha}(x-s)\right)\Bigr]dx\biggr|.
\end{eqnarray*}
Due to the L{\'e}vy-Khintchine formula \eqref{LK1},
we derive for any $u_{1},u_{2}\in\mathbb{R},$ 
\begin{multline*}
\left|\psi(u_{1}+u_{2})-\psi(u_{1})-\psi(u_{2})\right| \\
=\left|\sigma^{2}u_{1}u_{2}\right|
+\int_{0}^{\infty}\left|\exp(iu_{1}x)-1\right|\left|\exp(iu_{2})x)-1\right|\nu(dx)
\leq C \left|u_{1}\right|\left|u_{2}\right|
\end{multline*}
with \(C=\sigma^{2}+\int_{\mathbb{R}}x^{2}\nu(dx)<\infty.\)
As a result 
\begin{multline*}
\left|\mathrm{Cov_{\mathbb{C}}}\left[e^{iuZ_{t}},e^{ivZ_{s}}\right]\right|\\
\leq C \left|uv\right|\left(\left|\Phi_{(Z_{t},Z_{s})}(u,-v)\right|\vee\left|\Phi_{Z_{t}}(u)\Phi_{Z_{s}}(-v)\right|\right)\\
\times \int \K_{\alpha}(x)\K_{\alpha}(x+t-s)\,dx.
\end{multline*}
Hence 
\begin{eqnarray*}
\left|\mathrm{Cov_{\mathbb{C}}}\left[\Phi{}_{n}(u),\Phi{}_{n}(v)\right]\right| & = & \frac{1}{n^{2}}\sum_{j,k=1}^{n}\left|\mathrm{Cov_{\mathbb{C}}}\left(e^{iuZ_{j\Delta}},e^{ivZ_{k\Delta}}\right)\right|\\
 & \leq & \frac{\left|uv\right|}{n^{2}}\sum_{j,k=1}^{n}\left(\left|\Phi_{(Z_{j \Delta},Z_{k \Delta})}(u,-v)\right|\vee\left|\Phi(u)\Phi(-v)\right|\right)\\
 &  & \hspace{2.5cm}\times\int \K_{\alpha}(x)\K_{\alpha}(x+(k-j)\Delta)\,dx,
\end{eqnarray*}
where
\[
\int \K_{\alpha}(x)\K_{\alpha}(x+h)\,dx=0
\]
if $\left|h\right|>2/\alpha$ and
\[
\int \K_{\alpha}(x)\K_{\alpha}(x+h)\,dx\leq2\left(1-\alpha\left|h\right|/2\right)^{2/\alpha}
\]
for $\left|h\right|\leq2/\alpha.$ In the limiting case $\alpha\searrow0,$
we get 
\[
\int \K_{\alpha}(x)\K_{\alpha}(x+h)\,dx\leq2e^{-|h|}.
\]
As a result 
\begin{align*}
\left|\mathrm{Cov_{\mathbb{C}}}\left[\Phi{}_{n}(u),\Phi{}_{n}(v)\right]\right| & \leq\frac{Q(u,v)}{n^{2}}\sum_{0\leq\left|j-k\right|\leq\frac{2}{\Delta\alpha}}\left(1-\alpha\Delta\left|j-k\right|/2\right)^{2/\alpha}\\
 & \leq\frac{Q(u,v)}{n}\int_{0}^{\min\left\{ n,\frac{2}{\Delta\alpha}\right\} }\left(1-\alpha\Delta h/2\right)^{2/\alpha}\,dh\\
 & =\frac{Q(u,v)}{\Delta n}\int_{0}^{\min\left\{ \Delta n,\frac{2}{\alpha}\right\} }\left(1-\alpha r/2\right)^{2/\alpha}\,dr\\
 & =\frac{Q(u,v)}{\Delta n}\frac{2}{\alpha+2}\left(1-\min\left\{ \frac{\alpha\Delta n}{2},1\right\} \right)^{2/\alpha+1},
\end{align*}
where the function $Q(u,v)$ is bounded provided $\sigma>0.$ 

For further analysis of the terms \(I_1 - I_4\) in (\ref{eq:VarLn}), we need also some asymptotic upper bound for the relation \(|\Phi'(u) / \Phi(u)|\) for large \(u.\) Combining \eqref{Phiu} with \eqref{psiu}, we get
\[
\frac{\Phi'(u)}{\Phi(u)}=2\left( 
-\alpha\int_{0}^{1}\psi(yu)y{}^{\alpha-1}\,dy+\psi(u)
\right)/u,
\]
and therefore 
it holds $\left| \Phi'(u) / \Phi(u) \right|\lesssim
u$
as $|u|\to\infty$.

So we
have for $I_{1}$ 
\begin{align*}
I_{1} & \leq\frac{Q^{*}}{n}\left[\int_{0}^{U_{n}}\left|w_{\sigma}^{U_{n}}(u)\right|
\cdot
\left|
  \alpha - u \frac{\Phi'(u)}{\Phi(u)} 
\right|
\cdot
\left|\frac{1}{\Phi(u)}\right|\,du\right]^{2}\\
 & \leq\frac{Q^{*}}{n}\left[\int_{0}^{U_{n}}\left|U_{n}^{-3}w_{\sigma}^1(u/U_{n})\right|
 \left(
  \alpha + u \left|
           \frac{\Phi'(u)}{\Phi(u)} 
   \right|
\right)
 \left|\frac{1}{\Phi(u)}\right| du
 \right]^2\\
 & \leq\frac{Q^{*}}{nU_{n}}\left[\int_{\eps}^{1}\left|w_{\sigma}^{1}(u)u\right| \cdot
 \left|
	 \frac{\Phi'(u U_n)}{\Phi(u U_n)} 
\right| \cdot
 \left|
\frac{1}{\Phi(u U_n)}
\right|
\,du\right]^{2}\\
 & \lesssim\frac{1}{n\left|\Phi(U_{n})\right|^2}
\end{align*}
with some \(Q^*>0.\)
Analogously we get the upper bounds for \(I_2, I_3, I_4,\) for instance,
\[
I_{4}\leq\frac{Q^{*}}{n}\left[\int_{0}^{U_{n}}\left|w_{\sigma}^{U_{n}}(u) u \right|\frac{1}{\left|\Phi(u)\right|}\,du\right]^{2}
\lesssim 
\frac{1}{n U_n^2 \left|\Phi(U_{n})\right|^{2}},\quad n\to\infty.
\]

Finally, taking into account that 
\begin{eqnarray*}
|\Phi(U_n)| &=& 
\exp\left\{2\int_0^1 \Re [\psi(U_n y)]y^{\alpha-1}dy\right\} \\
&=&
\exp\left\{
-\frac{
\sigma^2 U_n^2
}{
 2+\alpha
}
-
\frac{
2 \lambda
}{
\alpha
}
+
2\int_0^1 \Re \F[\nu](U_n y) y^{\alpha-1} dy
\right\}
\gtrsim
e^{-C_4 U_n^2}
\end{eqnarray*}
with any \(C_4>\sigma^2/(4+2 \alpha),\) we conclude that  due to Markov inequality,
\begin{eqnarray}\label{Ln}
\P 
\left\{
|L_n| > 
A g_{n,2}
\right\}
\leq 
\frac{
1
}{
A^2
}, \qquad \mbox{where} \quad g_{n,2}=\frac{
e ^{(C_4/2) U_n^2}
}{
n^{1/2}
}.
\end{eqnarray}
\textbf{4.} Joint consideration of \eqref{bias}, \eqref{Rn} and \eqref{Ln} concludes the proof. In fact, under the choice \(U_n=\sqrt{\kappa \log(n)}\) with \(\kappa<\min\left(C_4^{-1}, (2 C_2 \sigma^2)^{-1} \right)\) 
we get that both  \(g_{n,1}\) and \(g_{n,2}\) are of the polynomial order.

\subsection{Proof of Theorem~\ref{thm2}}

Below we  focus on the proof of the first statement. 

A scheme for the proof of lower bounds is introduced in \cite{BR} and (more generally) in \cite{Tsyb}. Shortly speaking, it is sufficient to  construct two models  from the class \(\mathcal{T}_s(\sigma^\circ, R)\), say $(\gamma_0, \sigma_0,\nu_0)$ and $(\gamma_1, \sigma_1,\nu_1)$ (depending on \(n\)), such that \[|\sigma_1^2 - \sigma_0^2| \geq 2 A \log(n)^{-(s+3)/2},\] and the \(\chi^2\)-difference between the corresponding probability measures is bounded by \(1/n:\)
\begin{eqnarray*}
\chi^2 (p_1 | p_0) 
=
\int \frac{
\left(
  p_1 (x)  - p_0 (x)
\right)^2
}{
p_0(x)
}
dx \lesssim n^{-1},
\end{eqnarray*} 
where \(p_0\) and \(p_1\) are the probability densities for the first  and the second models resp.

\textbf{1.} Let us first present the models. The first model has the triplet  $(0,\sigma_0,\nu_0)\in\mathcal{T}_s(\sigma^\circ, R)$
with $\sigma_0=\sigma=\sigma^\circ/2$ and a L{\'e}vy density
\(
   \nu_0(x)=\nu(x) = c(1+\left|x\right|)^{-4},
\)
where \(c>0\) is chosen to guarantee \(\left\Vert \nu^{(s)}\right\Vert _{\infty}\leq R.\) We now perturb $(\sigma,\nu)$ such
that for low frequencies the characteristic functions still coincide. For this reason, we take  a flat-top kernel  $K$ such that 
\begin{eqnarray*}
\mathcal{F}K(u)=\begin{cases}
1, & |u|\leq1,\\
\exp\left(-\frac{e^{-1/(|u|-1)}}{2-|u|}\right), & 1<|u|<2,\\
0, & |u|\geq2.
\end{cases}
\end{eqnarray*}
This kernel and its derivatives have polynomial decay of any order, that is, for any \(r=0,1,2,...\) and any \(q=1,2,.. \) it holds \(K^{(r)}(x) \leq (1+|x|)^{-q}\) at least for large \(|x|\). Introduce $K_{h}(x)=h^{-1}K(h^{-1}x)$ for some (bandwidth) $h>0$.

Introduce the second model via the triplet \((0, \sigma_1,  \nu_1)\), where 
\[
\sigma_{1}^{2}=\sigma^{2}+2\delta, \qquad
\nu_{1}(x)=\nu(x)+\delta K_{h}''(x)
\]
with some \(\delta>0,\) which we will specify latter. 
Note that this model also belongs to the considered class 
\(\mathcal{T}_s(\sigma^\circ, R)\) when $h$ is small enough, provided
$\delta=o(h^{3})$ since then as $h\to0$ 
\[
\delta\left|K_{h}''(x)\right|=\delta h^{-3}\left|K''(x/h)\right|\lesssim\delta h^{-3}(1+\left|x\right|/h)^{-4}=o((1+\left|x\right|)^{-4})=o(\nu(x))
\]
(uniformly over $x\in\R$) follows from the polynomial decay of $K''$
of any order. 

\textbf{2.} On the second step, we consider the difference between the models.  For the corresponding characteristic exponents we obtain
(note ${\cal F}K_{h}''(u)=-u^{2}{\cal F}K_{h} (u)$, $\int K_{h}''(u) du=0$):
\[
\psi_{1}(u)-\psi_{0}(u)=\delta u^{2}(1-{\cal F}K(hu)),
\]
which is zero for $u\in[-h^{-1},h^{-1}]$. 

For further analysis, we need a lower bound for  the marginal density $p_{0}$
of the process 
\[
Z_{0,s}=\int \K_{\alpha}(s-t)\,dL_{0,t},
\]
where $L_{0,t}$ is a Levy process with triplet $(0,\sigma_{0},\nu_{0}).$ Note that since the process \(Z_s\) is stationary, we can take any \(s,\) in particular, \(s=0.\) Taking into account the decomposition~\eqref{model}, 
we conclude that 
\begin{align*}
p_{0}(x) & =\Big(N(0,\sigma_{\K}^{2})\ast\sum_{k=1}^{\infty}\frac{e^{-\lambda\Delta}(\lambda\Delta)^{k}}{k!}q_{k}\Big)(x),
\end{align*}
where $\sigma_{\K}^{2}=\sigma_{0}^{2}\int \K_{\alpha}^{2}(s)\,ds$ and
$q_{k}$ is the density of a random variable 
\[
\sum_{j=1}^{k}\K_{\alpha}\left(-U_{j}\right)\xi_{j}
\]
with $\xi_{1},\ldots,\xi_{k}$ being i.i.d random variables with density
$\nu_{0}/\lambda$ and $U_{1}<\ldots<U_{k}$ being i.i.d. random variables
with uniform law on $[-1/\alpha, 1/\alpha]$.
In view of the positivity of the summands, $\nu_{0}\gtrsim\nu$ and
the exponential decay of the Gaussian density (uniformly for $\Delta\lesssim1$
and keeping $\lambda,\sigma_{0},\nu_{0}$ fixed), we derive 
\[
p_{0}(x)\geq\lambda\Delta e^{-\lambda\Delta}(N(0,\sigma_{\K}^{2})\ast q_{1})(x) \gtrsim
\left(
  1 + |x|
\right)^{-4}.
\]
This yields the following upper bound for the \(\chi^2-\) difference between the models:
\begin{eqnarray*}
\chi^2 (p_1 | p_0) 
=
\int \frac{
\left(
  p_1 (x)  - p_0 (x)
\right)^2
}{
p_0(x)
}
dx 
\lesssim 
\int \left(
  1+|x|^4
\right) \left(
  p_1 (x)  - p_0 (x)
\right)^2
dx,
\end{eqnarray*}
and due to the Plancherel identity, we get 
\begin{eqnarray}
\label{chi2}
\chi^2 (p_1 | p_0) 
\lesssim
\bigl\|
\Phi_1 - \Phi_0
\bigr\|_{\L^2}^2
+
\bigl\|
\left(
  \Phi_1 - \Phi_0
 \right)''
\bigr\|_{\L^2}^2
\end{eqnarray}
 With the inequality $\left|1-e^{-z}\right|\le2\left|z\right|$ for
$z=x+iy\in\mathbb{C}$ with $x\ge0$ we can estimate the \(\L^2\)-norm between the characteristic functions \(\Phi_0\) and \(\Phi_1\):
\begin{align*}
\left\Vert \Phi_{1}-\Phi_{0}\right\Vert _{L^{2}}^{2} & \le\int2\max\left(\left|\Phi_{0}(u)\right|,\left|\Phi_{1}(u)\right|\right){}^{2} \left|\Psi_{1}(u)-\Psi_{0}(u)\right|^{2}du\\
 & \lesssim\int_{\left|u\right|>h^{-1}}e^{
- (\sigma^2/(\alpha + 2))
 u^2} \left|u^{-\alpha}\int_{0}^{u}\left(\psi_{1}(z)-\psi_{0}(z)\right)z{}^{\alpha-1}\,dz\right|^{2}du\\
 & \lesssim\delta^{2}\int_{\left|u\right|>h^{-1}}
 e^{
- (\sigma^2/(\alpha + 2))
 u^2}
 u^{-2\alpha}\left|\int_{0}^{u}(-1+{\cal F}\K(hz))z{}^{\alpha+1}\,dz\right|^{2}du\\
 & \lesssim\delta^{2}h^{-2(\alpha+2)}\int_{\left|u\right|>h^{-1}}e^{
- (\sigma^2/(\alpha + 2))
 u^2}u^{-2\alpha}\\
 &\hspace{5cm}\times\left|\int_{0}^{uh}(-1+{\cal F}\K(y))y{}^{\alpha+1}\,dy\right|^{2}du\\
 & \lesssim\delta^{2}h^{-5}\int_{\left|v\right|>1}
  e^{
- (\sigma^2/(\alpha + 2))
 (v/h)^2}
  v^{-2\alpha}\\
   &\hspace{5cm}\times\left|\int_{0}^{v}(-1+{\cal F}\K(y))y{}^{\alpha+1}\,dy\right|^{2}dv\\
   & \lesssim\delta^{2}h^{-5}
   e^{
- (\sigma^2/(\alpha + 2))
 (1/h)^2}.
\end{align*}
where we use that 
\begin{eqnarray*}
|\Phi(u)| &=& 
\exp\left\{ 2
u^{-\alpha} \int_0^u \left(
  -\frac{1}{2} \sigma^2 z^2 - \lambda+ \Re \F[\nu](z) 
\right)
z^{\alpha-1} dz
\right\}  \\
&\leq& \exp\left\{-( 2 \sigma^2/(2\alpha + 4))u^2\right\},
\end{eqnarray*}
since  \(| \Re \F[\nu](z) |< \lambda.\)
Analogously, we get the upper bound for the second summand in \eqref{chi2}:
\begin{multline*}
\left\Vert \left( 
\Phi_{1}-\Phi_{0}
\right)''\right\Vert _{L^{2}}^{2} 
\\= 
\int_{\left|u\right|>h^{-1}}
\Bigl|
\left(
\Psi''_1 (u) + 
\Psi'_1 (u)
\right)
\Phi_1(u) 
+
\left(
\Psi''_0 (u) + 
\Psi'_0 (u)
\right)
\Phi_0(u) 
\Bigr|^2 du\\
 \lesssim 
\int_{\left|u\right|>h^{-1}}
\Bigl|
\left(
\Psi''_1 (u) + 
\Psi'_1 (u)
\right)
+
\left(
\Psi''_0 (u) + 
\Psi'_0 (u)
\right)
\Bigr|^2 
e^{- (2 \sigma^2/(\alpha + 2))u^2}
du.
\end{multline*}
Due to the definition of the class (\(\int x^2 \nu(dx) <\infty\)) and to the assumption \(\nu(\R)<\infty\), we get that \(|\psi'(u)| \lesssim 1 + |u|\) and \(|\psi''(u)|\lesssim 1\) as \(u \to \infty.\) Therefore, applying \eqref{Phi}, we get the same asymptotics for the first and second derivatives of the function \(\Psi_0(u),\) whereas  \(\Psi'_1 (u) \lesssim 1+|u|+ \delta h^{-1}\) and  \(\Psi''_1 (u)  \lesssim 1+ \delta h^{-2}\). Finally, we get
\begin{eqnarray*}
\left\Vert \left( 
\Phi_{1}-\Phi_{0}
\right)''\right\Vert _{L^{2}}^{2}
   &\lesssim& \delta^{2}h^{-4}
   e^{
- (2 \sigma^2/(\alpha + 2))
 (1/h)^2}.
\end{eqnarray*}
\textbf{3.} To conclude the proof, we choose \(\delta= \delta' h^{s+3}\) with fixed \(\delta'\), and \[h=\left(
   2 \sigma^2 /(\alpha +2) 
\right)^{1/2} \left(
  \log(n)
\right)^{-1/2}.\]
Then
\begin{eqnarray*}
\chi^2 (p_1 | p_0)  \lesssim h^{2s+1}  e^{
- ( 2 \sigma^2/(\alpha + 2))
 (1/h)^2} \lesssim n^{-1},
\end{eqnarray*}
and 
\begin{eqnarray*}
| \sigma_1^2 - \sigma_0^2| = 2 \delta =  
C
\left(
  \log(n)
\right)^{-(s+3)/2}
\end{eqnarray*}
with some constant \(C\) depending on \(\alpha\) and \(\sigma^\circ\). This observation completes the proof.

\section*{Acknowledgment}
The study has been funded by the Russian Academic Excellence Project ``5-100''. 

\section*{Appendix. Some auxiliary results}
\begin{lem}\label{lem1}
Let \(\psi(u)\) be a characteristic exponent of \(L_t\) in the form \eqref{LK1}, and let \(\int |x| \nu(dx) < \infty\). Then for any \(\alpha>0\), it holds
\(
\lim_{u \to 0+}
\psi(u)u^{\alpha-1} = 0.
\)
\end{lem}
\begin{proof}
Note that
\begin{eqnarray*}
\lim_{u \to 0+} \frac{
\psi(u)
}{
u^{1-\alpha}
}
=
\lim_{u \to 0+} \frac{
\i \gamma u - (\sigma^2/2) u^2 
}{
	u^{1-\alpha}
}
+
\lim_{u \to 0+}
\left[
  u^\alpha
  \cdot
\int_{\R/\{0\}}
\frac{
  e^{\i u x}-1		
}{
  u
}\nu (dx) 
\right]=0,
\end{eqnarray*}
since 
\begin{eqnarray*}
\lim_{u \to 0+}
\int_{\R/\{0\}}
\frac{
  e^{\i u x}-1		
}{
  u
}\nu (dx)
=
\int_{\R/\{0\}}
\lim_{u \to 0+}
\frac{
  e^{\i u x}-1		
}{
  u
}\nu (dx) = \i \int_{\R/\{0\}} x \nu (dx),
\end{eqnarray*}
where the change of places between limit and integral is possible due to the Lebesque theorem. In fact, 
\begin{eqnarray*}
\left|
\frac{
  e^{\i u x}-1		
}{
  u
}
\right|
\leq 
\frac{
  1- \cos(u x)
}{
  u
}
+ 
\frac{
  |\sin(ux)|
}{
  u
}
\leq 2 |x|,
\end{eqnarray*}
and
 \(\int |x| \nu(dx) <\infty\) due to the assumption.
\end{proof}
%
%
%
\section*{References}
\bibliographystyle{plain}
\bibliography{bibliography-final-1}

\end{document}